\documentclass[notitlepage,11pt,onecolumn]{article}%
\usepackage{amsmath}
\usepackage{amsfonts}
\usepackage{aip}
\usepackage{amssymb}
\usepackage{graphicx}%
\providecommand{\U}[1]{\protect\rule{.1in}{.1in}}
\newtheorem{theorem}{Theorem}
\newtheorem{acknowledgement}[theorem]{Acknowledgement}

\newenvironment{proof}[1][Proof]{\noindent\textbf{#1.} }{\ \rule{0.5em}{0.5em}}
\begin{document}

\title{Symplectic Coarse-Grained Classical and Semclassical Evolution of Subsystems:
New Theoretical Approach}
\author{Maurice A. de Gosson\thanks{maurice.de.gosson@univie.ac.at}\\University of Vienna\\Faculty of Mathematics (NuHAG)\\Oskar-Morgenstern-Platz 1\\1090 Vienna AUSTRIA}
\maketitle

\begin{abstract}
We study the classical and semiclassical time evolutions of subsystems of a
Hamiltonian system; this is done using a generalization of Heller's thawed
Gaussian approximation introduced by Littlejohn. The key tool in our study is
an extension of Gromov's "principle of the symplectic camel" obtained in
collaboration with N. Dias and J. Prata. This extension says that the
orthogonal projection of a symplectic phase space ball on a phase space with a
smaller dimension also contains a symplectic ball with the same radius. In the
quantum case, the radii of these symplectic balls are taken equal to
$\sqrt{\hbar}$ and represent ellipsoids of minimum uncertainty, which we have
called "quantum blobs" in previous work.

\end{abstract}

\begin{center}

\end{center}

\section{Introduction}

Let us consider a bipartite physical system $A\cup B$ consisting of two
subsystems $A$ and $B$ with phase spaces $\mathbb{R}^{2n_{A}}$ and
$\mathbb{R}^{2n_{B}}$. We assume that both $A$ and $B$ are Hamiltonian, with
respective Hamiltonian functions $H_{A}(x_{A},p_{A})$ and $H_{B}(x_{B},p_{B}%
)$. As long as $A$ and $B$ do not interact in any way, the time evolution of
the total system $A\cup B$ can be predicted by solving separately the Hamilton
equations
\begin{equation}
\left\{
\begin{array}
[c]{c}%
\dot{x}_{A}=\nabla_{p_{A}}H_{A}(x_{A},p_{A})\\
\dot{p}_{A}=-\nabla_{x_{A}}H_{A}(x_{A},p_{A})
\end{array}
\right.  \text{\ } \label{A}%
\end{equation}
for the system $A$ and those,%
\begin{equation}
\left\{
\begin{array}
[c]{c}%
\dot{x}_{B}=\nabla_{p_{B}}H_{B}(x_{B},p_{B})\\
\dot{p}_{B}=-\nabla_{x_{B}}H_{B}(x_{B},p_{B})
\end{array}
\right.  \label{B}%
\end{equation}
for the system $B$. The composite system $A\cup B$ has Hamiltonian
$H=H_{A}+H_{B}$ defined on $\mathbb{R}^{2n}\equiv\mathbb{R}^{2n_{A}}%
\times\mathbb{R}^{2n_{B}}$ and its evolution is fully determined by the values
of $(x_{A},p_{A})$ and $(x_{B},p_{B})$ at some initial time, say, $t=0$. The
situation becomes much more intricate when the two systems $A$ and $B$ are
allowed to interact (which is generically the case). One must then add an
interaction term $H_{\mathrm{inter}}$ to $H_{A}+H_{B}$, and the total
Hamiltonian is then
\begin{multline}
H=H_{A}(x_{A},p_{A})+H_{B}(x_{B},p_{B})\label{Hab}\\
+H_{\mathrm{inter}}(x_{A},x_{B};p_{A},p_{B})~.
\end{multline}
Such Hamiltonian functions frequently appear in molecular dynamics and in the
Kepler problem. For instance, we can assume that the whole system consists of
$N$ molecules moving in physical three-dimensional space so that the full
phase space is ${\mathbb{R}}^{2n}$ with\ $n=3N$, and focus on a subset of
$N_{A}$ particles with phase space ${\mathbb{R}}^{2n_{A}}$, $n_{A}=3N_{A}$.
This subset is then viewed as a classical open system
\cite{BogoBogo,Bogo,Davies} interacting with its environment. The solutions to
Hamilton's equations for (\ref{Hab}) are generally in no way simply related to
the solutions (\ref{A}) and (\ref{B}) of the uncoupled problem (\ref{A}%
)--(\ref{B}), making their study usually very complicated. For instance, let
$z_{A,0}=(x_{A,0},p_{A,0})$ be an initial condition in $\mathbb{R}^{2n_{A}}$
for the equations (\ref{A}); this point is the projection of $z_{0}%
=(z_{A,0},z_{B,0})$ for any values of $z_{B,0}=(x_{B,0},p_{B,0})$. The point
$z_{0}$, taken as initial datum for the Hamilton equations for $H$, will in
general be projected to infinitely many bifurcating trajectories in
$\mathbb{R}^{2n_{A}}$ all starting from $z_{A,0}$: the solutions to the
equations (\ref{A}) will depend not only on the initial value $z_{A,0}$ but
are parametrized by those, $z_{B,0}$, of the system $B$: any change in the
system $B$ will affect the system $A$. The motion of a subsystem of a
Hamiltonian system is thus usually not Hamiltonian as soon as there are
interactions with its environment. The situation is similar in quantum
mechanics. In this case the Hamiltonian functions $H_{A}$ and $H_{B}$ are
replaced with their quantizations\footnote{The choice of a quantization scheme
is always somewhat arbitrary; to keep things simple we will only use in this
paper the usual Weyl quantization. This has many technical advantages, one of
them being that Weyl quantization is symplectically covariant under
conjugation with metaplectic operators.} $\widehat{H_{A}}=H_{A}(\widehat{x}%
_{A},\widehat{p}_{A})$, $\widehat{H_{B}}=H_{B}(\widehat{x}_{B},\widehat{p}%
_{B})$ and the Hamilton equations (\ref{A}) and (\ref{B}) are replaced with
the corresponding Schr\"{o}dinger equations%
\begin{align}
i\hbar\partial_{t}\psi_{A}  &  =\widehat{H_{A}}\psi_{A}\text{ \ , \ }\psi
_{A}(\cdot,0)=\psi_{A,0}\label{psiAB}\\
i\hbar\partial_{t}\psi_{B}  &  =\widehat{H_{B}}\psi_{B}\text{ \ , \ }\psi
_{B}(\cdot,0)=\psi_{B,0}%
\end{align}
where $\psi_{A}\in L^{2}(\mathbb{R}^{n_{A}})$ and $\psi_{B}\in L^{2}%
(\mathbb{R}^{n_{B}})$ describe the states of the quantized systems $A$ and $B$
at time $t$. As long as the subsystems $A$ and $B$ do not interact, the
evolution of the bipartite system $A\cup B$ is determined by the
Schr\"{o}dinger equation
\begin{equation}
i\hbar\partial_{t}\psi=(\widehat{H_{A}}+\widehat{H_{B}})\psi\text{ \ },\text{
\ }\psi_{0}=\psi_{A,0}\otimes\psi_{B,0}~.
\end{equation}
However, when $A$ and $B$ are allowed to interact, the evolution of $A\cup B$
is described by the complete Schr\"{o}dinger equation%
\begin{equation}
i\hbar\partial_{t}\psi=\widehat{H}\psi\text{ \ , \ }\psi(\cdot,0)=\psi_{0}%
\end{equation}
where the operator
\[
\widehat{H}=\widehat{H_{A}}+\widehat{H_{B}}+\widehat{H}_{\mathrm{inter}}%
\]
is the quantization of the total Hamiltonian (\ref{Hab}) and $\psi_{0}\in
L^{2}(\mathbb{R}^{n})$. Even when the initial wave function $\psi_{0}$ is a
tensor product $\psi_{A,0}\otimes\psi_{B,0}$ the solution $\psi$ at time $t$
will generally not be a tensor product, because the state $|\psi\rangle$ will
be entangled \cite{nemes,Thaller} and of the type
\[
\psi=\sum_{j}c_{j}\psi_{A,j}\otimes\psi_{B,j}%
\]
so it does not make sense to attribute to the subsystem $A$ a pure state; it
will rather evolve into a mixed state.

Let us briefly describe the strategy we will adopt to study the classical and
quantum motion of subsystems. We begin by coarse-graining the total phase
space ${\mathbb{R}}^{2n}$ by balls $B_{R}^{2n}(z_{0})$ with small radius $R$
and center $z_{0}$. The shadow (= orthogonal projection) of this ball on the
phase space $\mathbb{R}^{2n_{A}}$ of the subsystem $A$ is of course a ball
$B_{R}^{2n_{A}}(z_{0,A})$ with the same radius $R$ in $\mathbb{R}^{2n_{A}}$
and centered at the projection $z_{0,A}$ of $z_{0}$. The next step involves
replacing the total Hamiltonian $H$ in (\ref{Hab}) with a local approximation
$H_{0}$ for each $z_{0}$. This local Hamiltonian is obtained as follows: let
$z_{t}$ be the solution of the Hamilton equations for $H$ passing through
$z_{0}$ at time $t=0$. $H_{0}$ is then obtained by truncating the Taylor
series of $H(z)$ at $z=z_{t}$ and retaining only terms of order $\leq2$:
\[
H_{0}(z,t)=H(z_{t})+\nabla_{z}H(z_{t})(z-z_{t})+\frac{1}{2}H^{\prime\prime
}(z_{t})(z-z_{t})^{2}~.
\]
If we take $z_{0}$ as the initial condition for the Hamilton equations for
$H_{0}$. the exact solution is $z_{t}$. If we choose an initial point close to
$z_{0}$ we will get a (hopefully good) approximation to the exact solution
(intuitively the smaller the radius $R$ the better this approximation will
be). This procedure is what Littlejohn \cite{Littlejohn} calls the
\textquotedblleft nearby orbit approximation\textquotedblright; when applied
to the semiclassical case it is a generalization of Heller's thawed Gaussian
approximation \cite{heller,hellerbook}; also see Hepp \cite{Hepp}. Now, the
local Hamiltonian $H_{0}$ is a quadratic polynomial in the phase space
variables, and the corresponding Hamilton equations are thus linear so that
the local Hamiltonian flow $\Phi_{t}$ they generate can be expressed using
only phase space translations and linear symplectic transformations. Taking
for simplicity $z_{0}=0$, after time $t$ this flow $\Phi_{t}$ will thus have
deformed the initial ball $B_{R}^{2n}(0)$ into a phase space ellipsoid
\[
\Omega_{t}=\{z_{t}\}+S_{t}B_{R}^{2n}(0)
\]
where $S_{t}$ is a symplectic matrix in ${\mathbb{R}}^{2n}$. Now comes the
crucial point: very recent results \cite{DiGoPra19} in symplectic geometry
show that the shadow $\Omega_{A,t}$ of $\Omega_{t}$\ on $\mathbb{R}^{2n_{A}}$
is an ellipsoid containing a symplectic ball centered at $z_{A,t}$:
\begin{equation}
\Omega_{A,t}\supset z_{A,t}+S_{A,t}B_{R}^{2n_{A}}(0) \label{formula1}%
\end{equation}
where $S_{A,t}$ is a symplectic matrix in the smaller phase space
${\mathbb{R}}^{2n_{A}}$ and $z_{A,t}$ the projection of $z_{t}$. This striking
(and non-trivial) result follows from a generalization \cite{DiGoPra19} (in
the linear case) of Gromov's famous symplectic non-squeezing theorem
\cite{Gromov}.\ We sketch the proof of this result in Theorem \ref{Thm3}. The
quantum analogue of a symplectic ball is what we have called a
\textquotedblleft quantum blob\textquotedblright\ in earlier work
\cite{blobs}, and the Wigner formalism shows that quantum blobs are in
one-to-one correspondence with pure Gaussian states (sometimes called
generalized squeezed coherent states). This allows us to propagate these
Gaussians semiclassically using a Weyl quantization
\[
\widehat{H}_{0}=H(z_{t})+\nabla_{z}H(z_{t})(\widehat{z}-z_{t})+\frac{1}%
{2}H^{\prime\prime}(z_{t})(\widehat{z}-z_{t})^{2}%
\]
of the local Hamiltonian. The main result is then that a pure Gaussian state
in the subsystem $A$ will evolve into a mixed (Gaussian) which will be
described explicitly in Theorem \ref{ThmSemi} using the semiclassical
propagator which can be expressed using only displacement and metaplectic operators.

\subsubsection{\textbf{Notation and terminology }}

\paragraph{General notation}

The phase space variable will be $z=(z_{A},z_{B})$ with $z_{A}=(x_{A},p_{A})$
and $z_{B}=(x_{B},p_{B})$. We will also use the notation $z=z_{A}\oplus z_{B}$
and
\[
J=J_{A}\oplus J_{B}=%
\begin{pmatrix}
J_{A} & 0\\
0 & J_{B}%
\end{pmatrix}
\]
where $J_{A}$ (\textit{resp.} $J_{B}$) is the standard symplectic matrix on
${\mathbb{R}}^{2n_{A}}$ (\textit{resp}. ${\mathbb{R}}^{2n_{B}}$). They
correspond to the symplectic structures $\sigma(z,z^{\prime})=Jz\cdot
z^{\prime}$, $\sigma_{A}(z_{A},z_{A}^{\prime})=J_{A}z_{A}\cdot z_{A}^{\prime}%
$, and $\sigma_{B}(z_{B},z_{B}^{\prime})=J_{B}z_{B}\cdot z_{B}^{\prime}$,
respectively. The scalar product of two vectors $u,v\in\mathbb{R}^{m}$ is
written $u\cdot v$ or $uv$; if $A$ is a symmetric $m\times m$ matrix we use
the shorthand notation $Au\cdot u=Au^{2}$. The symplectic group of the
symplectic space $({\mathbb{R}}^{2n},\sigma)$ is denoted by
$\operatorname*{Sp}(n)$; similarly the symplectic groups of $({\mathbb{R}%
}^{2n_{A}},\sigma_{A})$ and $({\mathbb{R}}^{2n_{B}},\sigma_{B})$ are
$\operatorname*{Sp}(n_{A})$ and $\operatorname*{Sp}(n_{B})$, respectively.

We denote by $B_{R}^{2n}(z_{0})$ the open ball in $\mathbb{R}^{2n}$ with
radius $R$ and center $z_{0}$:%
\[
B_{R}^{2n}(z_{0})=\{z\in\mathbb{R}^{2n}:|z-z_{0}|<R\}~.
\]
When $z_{0}=0$ we write $B_{R}^{2n}(0)=B_{R}^{2n}$ and
\[
B_{R}^{2n}(z_{0})=\{z_{0}\}+B_{R}^{2n}~.
\]
The volume of $B_{R}^{2n}(z_{0})$ is
\[
\operatorname*{Vol}(B_{R}^{2n}(z_{0}))=\frac{(\pi R^{2})^{n}}{n!}~.
\]
We call the image $S(B_{R}^{2n}(z_{0}))$ of $B_{R}^{2n}(z_{0})$ by a linear
canonical transformation $S\in\operatorname*{Sp}(n)$ a \emph{symplectic ball}
with center $Sz_{0}$ and radius $R$. By Liouville's theorem \cite{Arnold}
$S(B_{R}^{2n}(z_{0}))$ and $B_{R}^{2n}(z_{0})$ have the same volume.

\paragraph{Symplectic terminology}

It is also appropriate to shortly discuss here some points of symplectic
terminology. Strictly speaking, the \textquotedblleft symplectic
camel\textquotedblright\ we refer to in this article is usually called
\textquotedblleft(symplectic) non-squeezing theorem\textquotedblright%
\ \cite{Gromov}. What is called \textquotedblleft symplectic
camel\textquotedblright\ in the symplectic topology community is the following
statement: let $X(r)$ be the open subset of $\mathbb{R}^{2n}$ consisting of
the union of all points $(x_{1},p_{1},...,,x_{1},p_{1})$ with $x_{1}\neq0$ and
the open ball $B_{R}^{2n}$ (it is the complement in $\mathbb{R}^{2n}$ of the
hyperplane $\{x_{1}=0\}$ with a hole of radius $r$ centered at the origin. If
$R>r$ then there is no one-parameter family $(f_{t})$ of symplectomorphisms (=
canonical transformations) of $\mathbb{R}^{2n}$ such that $f_{t}(B_{R}%
^{2n})\subset X(r)$ for every $t\in\lbrack0,1]$, $f_{0}(B_{R}^{2n}%
)\subset\{x_{1}>0\}$, and $f_{1}(B_{R}^{2n})\subset\{x_{1}<0\}$. For a
complete proof, see Allais \cite{Allais}; also \cite{duff}. In the present
paper we will slightly abuse this terminology by referring to the diverse
reformulations of Gromov's non-squeezing theorem as the \textquotedblleft
principle of the symplectic camel\textquotedblright\ \cite{FP,physreps}.

\section{The Extended Symplectic Camel Principle\label{sec1}}

\subsection{Statement and discussion\label{sec12}}

In 1985 the mathematician M. Gromov \cite{Gromov} proved, using the theory of
pseudo-holomorphic curves, the following remarkable and highly non-trivial result:

\begin{theorem}
[Gromov]\label{Thm1}Let $Z_{j}^{2n}(r)$ be the phase space cylinder defined by
$x_{j}^{2}+p_{j}^{2}\leq r^{2}$. There exists a canonical transformation
$\Phi$ of $\mathbb{R}^{2n}$ such that $\Phi(B_{R}^{2n})\subset Z_{j}^{2n}(r)$
if and only $R\leq r$.
\end{theorem}

It was reformulated by Gromov and Eliashberg \cite{EliGro} in the form we will
use in this paper:

\begin{theorem}
[Symplectic camel]\label{Thm2}Let $\Phi$ be a canonical transformation of
$\mathbb{R}^{2n}$ and $\Pi_{j}$ the orthogonal projection $\mathbb{R}%
^{2n}\longrightarrow\mathbb{R}_{x_{j},p_{j}}^{2}$ on any plane of conjugate
variables $x_{j},p_{j}$. We have
\begin{equation}
\operatorname*{Area}\Pi_{j}(\Phi(B_{R}^{2n}))\geq\pi R^{2}~. \label{area}%
\end{equation}

\end{theorem}

Note that Theorem \ref{Thm2} trivially implies Theorem \ref{Thm1} while the
converse implication follows from the fact that any planar domain of area
smaller than $\pi R^{2}$ can be mapped into a disk of same area by ana
area-preserving diffeomorphism.

These results at first sight seem to contradict the common conception of
Liouville's theorem on volume conservation; they are in fact
\emph{refinements} of it. The PSC has in fact the following dynamical
interpretation: assume that we are moving the ball $B_{R}^{2n}(z_{0})$ through
phase space using some Hamiltonian flow $\Phi_{t}$, which is a one-parameter
family of canonical transformations. In view of Liouville's theorem the
deformed ball $\Phi_{t}(B_{R}^{2n}(z_{0}))$ will have the same volume as
$B_{R}^{2n}(z_{0})$; the principle of the symplectic camel says that in
addition its \textquotedblleft shadow\textquotedblright\ (orthogonal
projection) on any $x_{j},p_{j}$ plane will never decrease below its initial
value $\pi R^{2}$. The result ceases to be true if we move the ball using
non-Hamiltonian volume-preserving flows: it is the symplectic character of
Hamiltonian flows which plays an essential role here. It does not require very
much imagination to realize that the PSC is reminiscent of the quantum
uncertainty principle. In fact, we have shown in \cite{Birk,FP,physreps} that
Heisenberg's uncertainty principle in its strong form (the
Robertson--Schr\"{o}dinger inequalities) can be concisely reformulated using
the PSC, and that this reformulation also extends to the case of classical
uncertainties \cite{jstat}. We mention that Kalogeropoulos \cite{kalo2} has
been able to use the principle of the symplectic camel to study a non-standard
characterization of thermodynamical entropy (also see the related paper
\cite{kalo1}).

The planes of conjugate coordinates $x_{j},p_{j}$ just considered are
particularly simple examples of phase subspaces of $\mathbb{R}^{2n}$ and
correspond to the case $n_{A}=1$ in the notation of the Introduction. A
natural question which arises is whether the PSC can be extended to symplectic
subspaces of higher dimension, \textit{i.e.} to arbitrary phase sub-spaces
$\mathbb{R}^{2n_{A}}$. For general nonlinear canonical transformations the
situation is not yet very well understood (this will be discussed in Section
\ref{secpersp}). However, in a recent work \cite{DiGoPra19} we have proved, in
collaboration with N. Dias and J. Prata, the following refinement of the PSC
for linear canonical transformations:

\begin{theorem}
[Extended symplectic camel]\label{Thm3}Let $\Pi_{A}$ be the orthogonal
projection $\mathbb{R}^{2n_{A}}\times\mathbb{R}^{2n_{B}}\longrightarrow
\mathbb{R}^{2n_{A}}$. There exists $S_{A}\in\operatorname*{Sp}(n_{A})$ such
that for every $R>0$ the projected ellipsoid $\Pi_{A}(S(B_{R}^{2n}))$ contains
the symplectic ball $S_{A}(B_{R}^{2n_{A}})$:%
\begin{equation}
\Pi_{A}(S(B_{R}^{2n}))\supset S_{A}(B_{R}^{2n_{A}}))~. \label{projo}%
\end{equation}
More generally,
\begin{equation}
\Pi_{A}(S(B_{R}^{2n}(z_{0})))\supset\{\Pi_{A}(Sz_{0})\}+S_{A}(B_{R}^{2n_{A}%
})~. \label{proja}%
\end{equation}
We have equality in (\ref{projo}), (\ref{proja}) if $S=S_{A}\oplus S_{B}$ for
some $S_{B}\in\operatorname*{Sp}(n_{B})$.
\end{theorem}

Note that (\ref{proja}) immediately follows from (\ref{projo}) so that it is
sufficient to focus our attention on the proof of the inclusion (\ref{projo}).
The idea of the proof goes as follows (see \cite{DiGoPra19} for details): the
ball $\Omega=S(B_{R}^{2n})$ is determined by the inequality $Pz^{2}\leq R^{2}$
where $P=(SS^{T})^{-1}$. Writing $P$ in block matrix form
\begin{equation}
P=%
\begin{pmatrix}
P_{AA} & P_{AB}\\
P_{BA} & P_{BB}%
\end{pmatrix}
\label{paa}%
\end{equation}
the blocks $P_{AA}$, $P_{AB}$, $P_{BA}$, $P_{BB}$ having dimensions
$2n_{A}\times2n_{A}$, $2n_{A}\times2n_{B}$, $2n_{B}\times2n_{A}$,
$2n_{B}\times2n_{B}$, respectively, the projection $\Omega_{A}=\Pi_{A}%
(S(B_{R}^{2n}))$ on ${\mathbb{R}}^{2n_{A}}$ is then the ellipsoid
\begin{equation}
\Omega_{A}=\{z_{A}:(P/P_{BB})z_{A}^{2}\leq R^{2}\} \label{ppbb}%
\end{equation}
where the $2n_{A}\times2n_{A}$ symmetric and positive definite matrix
$P/P_{BB}$ is the Schur complement of $P_{BB}$ in $P$, that is
\begin{equation}
P/P_{BB}=P_{AA}-P_{AB}P_{BB}^{-1}P_{BA}~. \label{Schur}%
\end{equation}
It follows that $\Omega_{A}$ is a non-degenerate ellipsoid in $\mathbb{R}%
^{2n_{A}}$. In view of Williamson's diagonalization theorem \cite{Birk} there
exists a symplectic matrix $S_{A}\in\operatorname*{Sp}(n_{A})$ diagonalizing
$P/P_{BB}$, that is
\begin{equation}
P/P_{BB}=(S_{A}^{-1})^{T}D_{A}S_{A}^{-1} \label{mmbsa}%
\end{equation}
where $D_{A}$ has the form
\begin{equation}
D_{A}=%
\begin{pmatrix}
\Lambda_{A} & 0\\
0 & \Lambda_{A}%
\end{pmatrix}
\label{da}%
\end{equation}
with $\Lambda_{A}=\operatorname*{diag}(\lambda_{1},...,\lambda_{n_{A}})$, the
positive numbers $\lambda_{j}$ being the symplectic eigenvalues of $P/P_{BB}$
(\textit{i.e.} the positive numbers $\lambda_{j}$ such that $\pm i\lambda_{j}$
is an eigenvalue of $J_{A}(P/P_{BB})$). The symplectic matrix $S_{A}$ in
(\ref{mmbsa}) is the one appearing in (\ref{projo}), (\ref{proja}), and one
proves that $D_{A}\leq I$, that is
\begin{equation}
\lambda_{j}\leq1\text{ \ \textit{for} \ }j=1,...,n_{A}~. \label{lambdaj}%
\end{equation}
The proof of the inclusion (\ref{projo}) now follows: in view of (\ref{mmbsa})
the inequality (\ref{ppbb}) is equivalent to
\[
(S_{A}^{-1})^{T}D_{A}S_{A}^{-1}z_{A}^{2}\leq R^{2}~;
\]
since $D_{A}\leq I$ it implies that the projection $\Omega_{A}$ must contain
the ball $\Omega_{A}=S_{A}(B_{R}^{2n_{A}})$. That the condition $S=S_{A}\oplus
S_{B}$ is sufficient to insure equality in (\ref{projo}) and (\ref{proja}) is
clear. It is however not necessary: Consider for instance the case where $S$
is a symplectic rotation. Then $SJ=JS$ and $S(B_{R}^{2n})=B_{R}^{2n}$ and the
equality in (\ref{projo}) holds if $S_{A}$ is the identity (or, more
generally, a symplectic rotation of $\mathbb{R}^{2n_{A}}$).

Note that since symplectic mappings are volume-preserving we have%
\begin{equation}
\operatorname*{Vol}\Pi_{A}(S(B_{R}^{2n}))\geq\operatorname*{Vol}B_{R}^{2n_{A}%
}~. \label{abbo1}%
\end{equation}
This inequality qualifies Theorem \ref{Thm3} as an extension of the principle
of the symplectic camel: since $S_{A}$ is volume preserving in ${\mathbb{R}%
}^{2n_{A}}$ it reduces to Theorem \ref{Thm2} when $n_{A}=1$ in the linear
case. The inequality (\ref{abbo1}) was in fact proved directly by Abbondandolo
and Matveyev \cite{abbo} some time ago, using methods from linear algebra. We
mention that Abbondandolo and Benedetti \cite{abbonew} have very recently
improved the inequality (\ref{abbo1}) by showing that a similar inequality
still holds for canonical transformations close to linear ones. It is an open
question whether Theorem \ref{Thm3} can be improved to encompass such
transformations (see the discussion in Section \ref{secpersp}).

\subsection{The nearby orbit method}

Consider the generalized time-dependent harmonic oscillator with Hamiltonian%
\begin{equation}
H(z,t)=\frac{1}{2}M(t)z^{2} \label{HM}%
\end{equation}
where $M=M(t)$ is a real symmetric $2n\times2n$ matrix depending continuously
on time $t$. The associated Hamilton equations are linear and hence a solution
$z(t)$ of the associated Hamilton equations satisfies $z(t)=S_{t}z(0)$ with
$S_{t}\in\operatorname*{Sp}(n)$. As an immediate application of the extended
principle of the symplectic camel (formula (\ref{proja}) in Theorem
\ref{Thm3}) there exists a one-parameter family of matrices $S_{A,t}%
\in\operatorname*{Sp}(n_{A})$ such that
\begin{equation}
\Pi_{A}(S_{t}(B_{R}^{2n}(z_{0})))\supset\{\Pi_{A}(S_{t}z_{0})\}+S_{A,t}%
(B_{R}^{2n_{A}})~. \label{pias}%
\end{equation}

Let now $H=H(z,t)$ be an arbitrary (possibly time-dependent) Hamiltonian
function on $\mathbb{R}^{2n}$. We assume $H$ to be at least twice continuously
differentiable in the position and momentum variables and once continuously
differentiable with respect to time $t$. Fixing a reference point $z_{0}$ in
phase space we denote by $z_{t}$ the solution to Hamilton's equation for $H$
with initial datum $z_{0}$ at time $t=0$. We will call the phase space curve
$t\longmapsto z_{t}$ the \textit{reference orbit\footnote{While Hamilton's
equations can usually not be solved exactly, there are efficient numerical
symplectic algorithms allowing to determine the reference orbit $z_{t}$ with
very good precision. See for instance \cite{channel,kang,xue,luo} and the
references therein.}.} We define the flow mapping $\Phi_{t}:\mathbb{R}%
^{2n}\longrightarrow\mathbb{R}^{2n}$ by $z(t)=\Phi_{t}(z(0))$ where $z(t)$ is
the solution of Hamilton's equations for the initial Hamiltonian $H$. The
mappings $\Phi_{t}$ are canonical transformations \cite{Arnold,Birk}. This
follows from the fact that the Jacobian matrix%
\begin{equation}
S_{t}(z_{0})=\partial(x_{t},p_{t})/\partial(x_{0},p_{0}) \label{Jacobian}%
\end{equation}
is symplectic, that is $S_{t}(z_{0})\in\operatorname*{Sp}(n)$ for all times
$t$. This property easily follows from the fact that $S_{t}(z_{0})$ satisfies
the \textquotedblleft variational equation\textquotedblright\ (see
\cite{Birk}, \S 2.3.2)%
\begin{equation}
\frac{d}{dt}S_{t}(z_{0})=JH^{\prime\prime}(z_{t},t)S_{t}(z_{0})
\label{variational}%
\end{equation}
where
\begin{equation}
H^{\prime\prime}=\left(  \partial^{2}H/\partial z_{\alpha}\partial z_{\beta
}\right)  _{1\leq\alpha,\beta\leq2n} \label{Hessian}%
\end{equation}
is the Hessian matrix of $H$, \textit{i.e.} the matrix of second derivatives
of $H$ in the phase space variables $z_{\alpha}=x_{\alpha}$ for $\alpha
=1,...,n$ and $z_{\alpha}=p_{\alpha}$ for $\alpha=n+1,...,2n$.

This leads us to consider the truncated Taylor expansion
\begin{equation}
H_{0}(z,t)=H(z_{t},t)+\nabla_{z}H(z_{t},t)(z-z_{t})+\frac{1}{2}H^{\prime
\prime}(z_{t},t)(z-z_{t})^{2} \label{ho}%
\end{equation}
of the original Hamiltonian $H$ around the point $z_{t}$. This new Hamiltonian
$H_{0}$ is time-dependent, even if $H$ is not. In the particular case where
$H$ has the simple physical form%
\begin{equation}
H(x,p,t)=\frac{1}{2}m^{-1}p^{2}+V(x,t) \label{hphys}%
\end{equation}
($m$ the mass matrix) the approximate Hamiltonian (\ref{ho}) takes the
familiar form \cite{begusic,Littlejohn}
\begin{equation}
H_{0}(x,p,t)=\frac{1}{2}m^{-1}(p-p_{t})^{2}+V_{\mathrm{LHA}}(x,t)
\label{kphys}%
\end{equation}
where $V_{\mathrm{LHA}}(x,t)$ is the time-dependent local harmonic
approximation of the potential $V(x)$:%
\begin{equation}
V_{\mathrm{LHA}}(x,t)=V(x_{t})+\nabla_{x}V(x_{t})(x-x_{t})+\frac{1}%
{2}V^{\prime\prime}(x_{t})(x-x_{t})^{2}~. \label{lha}%
\end{equation}
The solutions to the Hamilton equations for $H$ and $H_{0}$ coincide when the
initial value of $z(t)$ is chosen equal to the reference point $z_{0}$ for
both systems. In fact, the Hamilton equations for $H_{0}$ are
\begin{equation}
\dot{u}_{t}=J\nabla_{z}H(z_{t},t)+JH^{\prime\prime}(z_{t},t)(u_{t}-z_{t})
\label{hameq3}%
\end{equation}
and replacing $u_{t}$ with $z_{t}$ yields $\dot{z}_{t}=J\nabla_{z}H(z_{t},t)$.
Now, the solution of the Hamilton equations (\ref{hameq3}) with initial datum
$u_{0}$ is easily calculated and one finds that
\begin{equation}
u_{t}=z_{t}+S_{t}(z_{0})(u_{0}-z_{0})~. \label{ut}%
\end{equation}
Setting $u_{t}=U_{t}(z_{0})(u_{0})$ and introducing the phase space
translations $T(u):z\longmapsto z+u$ this formula can be rewritten as
\begin{equation}
U_{t}(z_{0})=T(z_{t})S_{t}(z_{0})T(-z_{0}) \label{phaseflow2}%
\end{equation}
so that $U_{t}(z_{0})$ can be viewed as a classical propagator
\cite{Littlejohn}. Recalling that we have set $z(t)=\Phi_{t}(z(0))$ the
discussion above shows that $\Phi_{t}(z)=U_{t}(z_{0})(z)$ when $z=z_{0}$,
which suggests to approximate the flow $\Phi_{t}$ by the affine mappings
$U_{t}(z_{0})$ for points close to $z_{0}$: this is the idea of the
\textquotedblleft nearby orbit approximation\textquotedblright\ with respect
to reference orbit\ $t\longmapsto z_{t}$. The validity of this approximation
can be tested using the standard theory of systems of differential equations
using the equality (\ref{ut}); the Lyapunov exponents are crucial for the
study of the accuracy of the solutions. In the absence of chaotic behavior it
is actually quite good for short times (Miller \cite{Miller}) which makes it
work well for low and or medium resolution electronic spectra. The following
straightforward consequence of the extended symplectic camel principle is new;
it describes the approximate motion of the projection on $\mathbb{R}^{2n_{A}}$
of $\Phi_{t}(B_{R}^{2n}(z_{0}))$:

\begin{theorem}
\label{ThmFund}Let $U_{t}(z_{0})$ be the flow determined by $H$ in the nearby
orbit approximation with reference orbit $t\longmapsto z_{t}$ starting from
$z_{0}$. The orthogonal projection%
\[
\Omega_{A,t}(z_{0})=\Pi_{A}(U_{t}(z_{0})(B_{R}^{2n}(z_{0})))
\]
contains a symplectic ball centered at $z_{A,t}=\Pi_{A}(z_{t})$:%
\begin{equation}
\Omega_{A,t}(z_{0})\supset\{z_{A,t}\}+S_{A,t}(z_{0})(B_{R}^{2n_{A}})
\label{satt}%
\end{equation}
where $S_{A,t}(z_{0})\in\operatorname*{Sp}(n_{A})$.
\end{theorem}

\begin{proof}
Since $T(-z_{0})B_{R}^{2n}(z_{0})=B_{R}^{2n_{A}}$ formula (\ref{phaseflow2})
yields
\begin{align*}
U_{t}(z_{0})(B_{R}^{2n}(z_{0}))  &  =T(z_{t})S_{t}(z_{0})(B^{2n}(0,R))\\
&  =\{z_{t}\}+S_{t}(z_{0})(B_{R}^{2n_{A})}~.
\end{align*}
Formula (\ref{satt}) now follows from the inclusion (\ref{pias}) taking into
account the linearity of the projection $\Pi_{A}$.
\end{proof}

Formula (\ref{satt}) implies the following important property of the subsystem
$A$: while the motion of the \textquotedblleft shadow\textquotedblright\ of
the initial ball $B_{R}^{2n}(z_{0})$ on $\mathbb{R}^{2n_{A}}$ cannot be
Hamiltonian (it is not volume preserving), it however contains a symplectic
ball whose evolution is governed by a Hamiltonian flow, namely that determined
by \
\begin{multline*}
H_{A}(z_{A},t)=-\frac{1}{2}J_{A}\dot{S}_{A,t}(z_{0})S_{A,t}(z_{0})^{-1}%
(z_{A}-z_{A,t})^{2}\\
+J_{A}z_{A}\cdot\dot{z}_{A,t}%
\end{multline*}
where $\dot{S}_{A,t}(z_{0})=\frac{d}{dt}S_{A,t}(z_{0})$. In fact,
\begin{multline*}
J_{A}\nabla_{z_{A}}H_{A}(z_{A,t},t)=\\
\dot{S}_{A,t}(z_{0})S_{A,t}(z_{0})^{-1}(z_{A}-z_{A,t})+\dot{z}_{A,t}%
\end{multline*}
and hence the solution $z_{A}(t)$ of the Hamilton equations for $H_{A}$
satisfies the linear differential equation
\begin{multline*}
\frac{d}{dt}(z_{A}(t)-z_{A,t})=\\
\dot{S}_{A,t}(z_{0})S_{A,t}(z_{0})^{-1}(z_{A}(t)-z_{A,t})~.
\end{multline*}
The solution of this equation is $z_{A}(t)=z_{A,t}+S_{A,t}(z_{0})$ and hence
our claim\footnote{On a more fundamental level this result is a consequence of
the fact that any smooth family $S_{t}$ of symplectic matrices such that
$S_{0}=I$ is the flow determined by some quadratic Hamiltonian \cite{RMP}.}.

\subsection{Entropy increase in subsystems\label{secentropy}}

In the nearby orbit approximation the orthogonal projection $\Omega
_{A,t}(z_{0})$ of $U_{t}(z_{0})(B_{R}^{2n}(z_{0}))$ is the ellipsoid defined
by the inequality
\begin{equation}
(P_{t}(z_{0})/P_{BB,t}(z_{0}))(z_{A}-z_{A,t})^{2}\leq R^{2} \label{sato}%
\end{equation}
where we have written the symplectic matrix
\[
P_{t}(z_{0})=(S_{A,t}(z_{0})S_{A,t}(z_{0})^{T})^{-1}%
\]
in block matrix form
\begin{equation}
P_{t}(z_{0})=%
\begin{pmatrix}
P_{AA,t}(z_{0}) & P_{AB,t}(z_{0})\\
P_{BA,t}(z_{0}) & P_{BB,t}(z_{0})
\end{pmatrix}
~. \label{blockp}%
\end{equation}
The volume of $\Omega_{A,t}(z_{0})$ is thus%
\[
\operatorname*{Vol}\Omega_{A,t}(z_{0})=\det(P_{t}(z_{0})/P_{BB,t}%
(z_{0}))\operatorname*{Vol}B_{R}^{2n_{A}}~.
\]
Recalling that the Schur complement satisfies the relation \cite{Zhang}:
\begin{equation}
\det P_{t}(z_{0})=\det(P_{t}(z_{0})/P_{BB,t}(z_{0}))\det P_{BB,t}(z_{0})~;
\label{Schur2}%
\end{equation}
in the present case $P_{t}(z_{0})$ is symplectic so that $\det P_{t}(z_{0}%
)=1$, and we thus have%
\begin{equation}
\det(P_{t}(z_{0})/P_{BB,t}(z_{0}))\det P_{BB,t}(z_{0})=1 \label{Schur3}%
\end{equation}
and hence%
\begin{equation}
\operatorname*{Vol}\Omega_{A,t}(z_{0})=\frac{1}{\det P_{BB,t}(z_{0}%
)}\operatorname*{Vol}B_{R}^{2n_{A}}~. \label{volate}%
\end{equation}
The entropy increase\footnote{That the entropy increases already follows from
the estimates in Abbondandolo and Matveyev \cite{abbo} on the volume of the
projection of a symplectic ball.} is thus%
\begin{equation}
\Delta\mathcal{S}=-k_{\mathrm{B}}\ln(\det P_{BB,t}(z_{0})) \label{delta}%
\end{equation}
where $\ln$ is the natural logarithm. This increase is due only to the
subsystem $B$. If $A$ and $B$ are uncoupled, then the cross-term
$P_{AB,t}(z_{0})=P_{BA,t}(z_{0})^{T}=0$ so that both $P_{AA,t}(z_{0})$ and
$P_{BB,t}(z_{0})$ are both symplectic and we have
\[
P_{t}(z_{0})=P_{AA,t}(z_{0})\oplus P_{BB,t}(z_{0})
\]
so that $\operatorname*{Vol}\Omega_{A,t}=\operatorname*{Vol}B_{R}^{2n_{A}}$
remains constant. This phenomenon shows the hardly surprising fact that the
Boltzmann entropy of the subsystem $A$ can take arbitrarily large values as
the subsystem $A$ occupies an increasing volume of phase space due to
interaction with the subsystem $B$.

There is another way to express the results above using the symplectic
eigenvalues $\lambda_{j,t}(z_{0})$ of $P_{BB,t}(z_{0})$. In view of formula
(\ref{mmbsa}) we have%
\[
P_{t}(z_{0})/P_{BB,t}(z_{0})=(S_{A,t}(z_{0})^{-1})^{T}D_{A,t}(z_{0}%
)S_{A,t}(z_{0})^{-1}%
\]
with
\begin{align}
D_{A,t}(z_{0})  &  =%
\begin{pmatrix}
\Lambda_{A,t}(z_{0}) & 0\\
0 & \Lambda_{A,t}(z_{0})
\end{pmatrix}
\label{date}\\
\Lambda_{A,t}(z_{0})  &  =\operatorname*{diag}(\lambda_{1,t}(z_{0}%
),...,\lambda_{n_{A},t}(z_{0}))~.
\end{align}
It follows that the volume of $\Omega_{A,t}$ is
\begin{align*}
\operatorname*{Vol}\Omega_{A,t}(z_{0})  &  =\det(D_{A,t}(z_{0}))^{-1/2}%
\operatorname*{Vol}B_{R}^{2n_{A}}\\
&  =\frac{1}{\lambda_{1,t}(z_{0})\cdot\cdot\cdot\lambda_{n_{A},t}(z_{0}%
)}\operatorname*{Vol}B_{R}^{2n_{A}}~.
\end{align*}
This volume can become arbitrarily large regardless of the dimension $n_{B}$
of the system $B$. For this it suffices that at least one of the symplectic
eigenvalues $\lambda_{j,t}(z_{0})$ is sufficiently small. The entropy increase
can thus be expressed as
\begin{equation}
\Delta\mathcal{S}=-k_{\mathrm{B}}\sum_{j=1}^{n_{A}}\ln\lambda_{j,t}(z_{0})~.
\label{entropy}%
\end{equation}

\subsection{Entropy and symplectic capacities}

One should however be aware of the fact that symplectic topology teaches us,
\textit{via }Gromov's non-squeezing theorem, that in Hamiltonian dynamics the
true measure of spreading is not volume, but \textit{symplectic capacity
}\cite{FP,physreps}. Symplectic capacities were introduced by Ekeland and
Hofer \cite{ekhof1,ekhof2}. A (normalized) symplectic capacity on
$(\mathbb{R}^{2n},\sigma)$ associates to every subset $\Omega$ of
$\mathbb{R}^{2n}$ a number $c(\Omega)\in\mathbb{R}_{+}\cup\{+\infty\}$ such
that the following properties hold \cite{ekhof1,ekhof2}:

\begin{itemize}
\item \textit{Monotonicity}: If $\Omega\subset\Omega^{\prime}$ then
$c(\Omega)\leq c(\Omega^{\prime})$;

\item \textit{Conformality}: For every real scalar $\lambda$ we have
$c(\lambda\Omega)=\lambda^{2}c(\Omega)$;

\item \textit{Symplectic invariance}: We have $c(f(\Omega))=c(\Omega)$ for
every canonical transformation $f$ of $\mathbb{R}^{2n}$;

\item \textit{Normalization}: We have
\begin{equation}
c(B_{R}^{2n})=\pi R^{2}=c(Z_{j,R}^{2n}) \label{cbz}%
\end{equation}
where $Z_{j,R}^{2n}$ is the cylinder $\{(x,p):x_{j}^{2}+p_{j}^{2}\leq R^{2}\}$.
\end{itemize}

That symplectic capacities exist is a consequence of Gromov's non-squeezing
theorem. The symplectic capacities $c_{\min}$ and $c_{\max}$ are defined by%
\begin{subequations}
\begin{align}
c_{\min}(\Omega)  &  =\sup_{f}\{\pi R^{2}:f(B^{2n}(R))\subset\Omega
\}\label{cmin}\\
c_{\max}(\Omega)  &  =\inf_{f}\{\pi R^{2}:f(\Omega)\subset Z_{j}^{2n}(R)
\label{cmax}%
\end{align}
where $f$ ranges over the set of all canonical transformations of
$\mathbb{R}^{2n}$. The symplectic capacity $c_{\min}$ is called the
\textquotedblleft Gromov width\textquotedblright\ while $c_{\max}$ is the
\textquotedblleft cylindrical capacity\textquotedblright. The notation is
motivated by the fact that every symplectic capacity $c$ on $(\mathbb{R}%
_{z}^{2n},\omega)$ is such that
\end{subequations}
\begin{equation}
c_{\min}(\Omega)\leq c(\Omega)\leq c_{\max}(\Omega) \label{cminmax}%
\end{equation}
for all $\Omega\subset\mathbb{R}^{2n}$. One also uses the \textit{linear
symplectic capacities}%
\begin{subequations}
\begin{align}
c_{\min}^{\mathrm{lin}}(\Omega)  &  =\sup_{S\in\operatorname*{Sp}(n)}\{\pi
R^{2}:f(B^{2n}(R))\subset\Omega\}\\
c_{\max}^{\mathrm{lin}}(\Omega)  &  =\inf_{S\in\operatorname*{Sp}(n)}\{\pi
R^{2}:f(\Omega)\subset Z_{j}^{2n}(R)
\end{align}
The symplectic capacity of an unbounded set can be finite (this is the case
for any unbounded set $\Omega$ such that $B_{R}^{2n}\subset\Omega\subset
Z_{j,R}^{2n})$, which shows that the notion of symplectic capacity is very
different from that of volume (except in the case $n=1$ where it is
essentially an \textit{area}: see our discussion in \cite{FP,physreps}).A
remarkable property is that all symplectic capacities agree on ellipsoids: if
\end{subequations}
\begin{equation}
\Omega=\{z\in\mathbb{R}^{2n}:Mz^{2}\leq R^{2}\} \label{CM}%
\end{equation}
where $M=M^{T}>0$, then for every symplectic capacity $c$ on $(\mathbb{R}%
^{2n},\sigma)$ we have
\begin{equation}
c(\Omega)=c_{\min}^{\mathrm{lin}}(\Omega)=c_{\max}^{\mathrm{lin}}(\Omega)=\pi
R^{2}/\lambda_{\max} \label{capomega}%
\end{equation}
where $\lambda_{\max}^{\sigma}$ is the largest symplectic eigenvalue of $M$.
The symplectic eigenvalues $\lambda_{1}^{\sigma},...,\lambda_{n}^{\sigma}$ of
$M$ are the numbers $\lambda_{j}^{\sigma}>0$ defined by the condition
\textquotedblleft\ $\pm i\lambda_{j}^{\sigma}$ is an eigenvalue of $JM$
\textquotedblright. Notice that $c(\Omega)$ does not depend on the dimension
of the ambient phase space, it is thus an \emph{extrinsic quantity}, as
opposed to volume.

Let us discuss our results on entropy from this new point of view. To
illustrate this, let us calculate $c(\Omega_{A,t})$. In view of formula
(\ref{satt}) and taking into account the translational invariance \cite{Birk}
of the symplectic capacities formula (\ref{capomega}) yields%
\[
c(\Omega_{A,t})=\pi R^{2}/\lambda_{\max,t}(z_{0})
\]
for every symplectic capacity $c$ where
\[
\lambda_{\max,t}(z_{0})=\max\nolimits_{j}\{\lambda_{j,t}(z_{0}),1\leq j\leq
n_{A}\}~.
\]

\section{The Semiclassical Case}

We denote by $\widehat{T}(z_{0})=e^{-i\sigma(\widehat{z},z_{0})/\hbar}$ the
Heisenberg displacement operator on $L^{2}(\mathbb{R}^{n})$. It is explicitly
given by the formula \cite{Birk,Littlejohn}%
\begin{equation}
\widehat{T}(z_{0})\psi(x)=e^{i(p_{0}x-p_{0}x_{0}/2)/\hbar}\psi(x-x_{0})~.
\label{heiwe}%
\end{equation}
We recall\ \cite{Birk,RMP,Littlejohn} that the symplectic group
$\operatorname*{Sp}(n)$ has a two-fold covering whose elements are unitary
operators acting on $L^{2}(\mathbb{R}^{n})$. This group is called the
metaplectic group $\operatorname*{Mp}(n)$. To every $S\in\operatorname*{Sp}%
(n)$ corresponds exactly two metaplectic operators $\pm\widehat{S}%
\in\operatorname*{Mp}(n)$ and we have the following intertwining property:
\begin{equation}
\widehat{T}(z)\widehat{S}=\widehat{S}\widehat{T}(S^{-1}z) \label{symco}%
\end{equation}
which is the analogue at the operator level of the obvious relation
\[
T(z)S=ST(S^{-1}z)~.
\]

\subsection{A generalization of the thawed Gaussian approximation}

We now consider the quantized version $\widehat{H}$ of the Hamiltonian $H$. We
will approximate $\widehat{H}$ by quantizing the nearby-orbit Hamiltonian
(\ref{ho}), which yields the operator
\begin{equation}
\widehat{H}_{0}=H(z_{t},t)+\nabla_{z}H(z_{t},t)(\widehat{z}-z_{t})+\tfrac
{1}{2}H^{\prime\prime}(z_{t},t)(\widehat{z}-z_{t})^{2}%
\end{equation}
where $\widehat{z}=(\widehat{x},\widehat{p})$ with $\widehat{x}=(\widehat{x}%
_{1},...,\widehat{x}_{n})$ ($\widehat{x}_{j}$ multiplication by $x_{j}$) and
$\widehat{p}=-i\hbar\nabla_{x}$. When $H$ has the simple physical form
(\ref{hphys}) this reduces to the simple operator
\begin{equation}
\widehat{H}_{0}=\frac{1}{2}M^{-1}(\widehat{p}-p_{t})^{2}+V_{\mathrm{LHA}%
}(\widehat{x},t) \label{klha}%
\end{equation}
where $V_{\mathrm{LHA}}(x,t)$ is given by (\ref{lha}). Let now $\psi_{0}$ be a
wavepacket (for instance, but not necessarily, a Gaussian). Assuming that
$\psi$ is well-localized around $x_{0}$ and its Fourier
transform\footnote{This supplementary condition is often forgotten in
practice; it is equivalent to saying that the Wigner transform of $\psi$ is
concentrated near $z_{0}=(x_{0},p_{0})$.} around $p_{0}$ one postulates that a
good approximation to the solution of the full Schr\"{o}dinger equation
\[
i\hbar\partial_{t}\psi=\widehat{H}\psi\text{ \ },\text{ \ }\psi(\cdot
,0)=\psi_{0}%
\]
is obtained by replacing $\widehat{H}$ with its approximation $\widehat{H_{0}%
}$. It is not difficult (Littlejohn \cite{Littlejohn}, \S 7) to see that the
approximate solution is then given by the quantum analogue of the equivalent
formulas (\ref{ut}) and (\ref{phaseflow2}):%
\begin{equation}
\psi(x,t)=e^{\frac{i}{\hbar}\gamma(t)}\widehat{T}(z_{t})\widehat{S}_{t}%
(z_{0})\widehat{T}(-z_{0})\psi_{0}(x) \label{st}%
\end{equation}
where $\gamma(t)$ is a phase correction, and $\widehat{S}_{t}(z_{0})$ the
metaplectic lift of the path $S_{t}(z_{0})$. We will call the mapping
\begin{equation}
\widehat{U}_{t}(z_{0})=e^{\frac{i}{\hbar}\gamma(t)}\widehat{T}(z_{t}%
)\widehat{S}_{t}(z_{0})\widehat{T}(-z_{0}) \label{stu}%
\end{equation}
the \textit{semiclassical propagator relative to the reference orbit} $z_{t}$.
This formula is interpreted as follows \cite{Birk,RMP}: let $S_{t}$ be the
phase space flow determined by a quadratic Hamiltonian of the type (\ref{HM}),
that is $H(z,t)=\frac{1}{2}M(t)z^{2}$. According to general principles from
the theory of covering spaces, this one-parameter family of symplectic
matrices can be lifted in a unique way to a one-parameter family of operators
$\widehat{S}_{t}$ in $\operatorname*{Mp}(n)$ such that $\widehat{S}_{0}=I_{d}%
$. The lifting is constructed as follows: since $\operatorname*{Mp}(n)$ is a
double covering of $\operatorname*{Sp}(n)$ to each symplectic matrix $S_{t}$
corresponds two metaplectic operators; for each time $t$ one then chooses the
operator leading to a continuous path $\widehat{S}_{t}$ passing through the
identity of $\operatorname*{Mp}(n)$ at time $t=0$; for details of the
construction, see \cite{Birk}, especially \S 7.2.2. This being done, one then
shows (\textit{ibid.}) that for any square integrable initial wavepacket
$\psi_{0}$ the function $\psi(x,t)=\widehat{S}_{t}\psi_{0}(x)$ is a solution
of the Schr\"{o}dinger equation with Hamiltonian operator $\widehat{H}%
=\frac{1}{2}M(t)\widehat{z}^{2}$. Formula (\ref{stu}) follows applying that
lifting principle to the classical flow (\ref{phaseflow2}) noting that
following a similar argument the one-parameter family $T(z_{t})$ (which is the
flow of the displacement Hamiltonian $H_{t}(z)=px_{t}-p_{t}x$) lifts to the
one-parameter family $\widehat{T}(z_{t})$ (which is the propagator for the
Hamiltonian operator $\widehat{H}_{t}=\widehat{p}x_{t}-p_{t}\widehat{x}$).

As already noted by Littlejohn \cite{Littlejohn} this construction is
essentially that of Heller and his collaborators
\cite{heller,he81,he05,hesuta82,he86} (also see Heller's recent monograph
\cite{hellerbook}) known as the \textquotedblleft thawed Gaussian
approximation\textquotedblright\ when the Hamiltonian is of the physical type
\textquotedblleft kinetic energy plus potential\textquotedblright. The
difference is that Heller built the time evolution into the parameters of a
Gaussian wave packet, while we have placed it into the operator, which has the
advantage that the initial wavefunction need not be Gaussian, and allows much
more general Hamiltonians. Such approximations (and their generalizations to
higher orders) have been extensively studied in physics and mathematics; a
non-exhaustive list of related papers is
\cite{CoNi,comb2,dahe84,grhe02,hado1,hado2,latr14,lu,Schubert,Straeng}. They
have applications to on-the-fly \textit{ab initio} semiclassical calculations
of molecular spectra \cite{begusic,wehrle1,wehrle2}; also see the recent paper
\cite{patoz} by Patoz \textit{et al}. For a very recent and up-to-date survey
with applications to computable algorithms see Lasser and Lubich
\cite{Lasser}; their paper in addition contains rigorous error estimates.

We mention that in a recent work \cite{begusic} Begu\v{s}ic \textit{et al}.
have considered a simplified variant of Heller's approximation obtained by
\textquotedblleft freezing\textquotedblright\ the Hessian $V^{\prime\prime
}(x_{t})$ of the potential in (\ref{lha}) at the initial position $x_{0}$; the
authors argue that this decreases the computational complexity, but at the
cost of a loss of accuracy. In our generalized context this would amount to
replacing the approximate Hamiltonian $H$ in (\ref{ho}) with
\begin{multline}
K_{0}(z,t)=H(z_{t},t)+\nabla_{z}H(z_{t},t)(z-z_{t})\\
+\frac{1}{2}H^{\prime\prime}(z_{0},t)(z-z_{t})^{2}%
\end{multline}
and the use of its quantized version $\widehat{K}_{0}$ to propagate
wavepackets. Our arguments still apply \textit{mutatis mutandis} if one uses
$K_{0}$ instead of $H_{0}$ since $K_{0}$ also is quadratic in the position and
momentum variables and thus allows the use of the metaplectic machinery.

\subsection{Gaussian mixed states and the Wigner ellipsoid}

It is well-known that there is a one-to-one correspondence between minimum
uncertainty phase space ellipsoids and Gaussian states; perhaps one of the
first systematic studies of this correspondence is Littlejohn's seminal paper
\cite{Littlejohn}. We have exploited this property in
\cite{Birk,FP,jstat,physreps} using the properties of the Wigner transform.
Working in canonical global coordinates $z=(x,p)$ the Wigner transform of
$\psi\in L^{2}(\mathbb{R}^{n})$ is, by definition,
\begin{equation}
W\psi(x,p)=\left(  \tfrac{1}{2\pi\hbar}\right)  ^{n}\int_{\mathbb{R}^{n}%
}e^{-\tfrac{i}{\hbar}p\cdot y}\psi(x+\tfrac{1}{2}y)\psi^{\ast}(x-\tfrac{1}%
{2}y)dy~. \label{wimo}%
\end{equation}
It satisfies the following transformation formulas
\cite{Birk,WIGNER,Littlejohn}%
\begin{align}
W(\widehat{S}\psi)(z)  &  =W\psi(S^{-1}z)\label{cov1}\\
W(\widehat{T}(z_{0})\psi)(z)  &  =W\psi(z-z_{0})~. \label{cov2}%
\end{align}

Here is a simple but fundamental example: let $\psi=\phi_{0}$ be the
standard\footnote{Littlejohn \cite{Littlejohn} calls it the \textquotedblleft
fiducial coherent state\textquotedblright.} centered coherent state:%
\begin{equation}
\phi_{0}(x)=(\pi\hbar)^{-n/4}e^{-|x|^{2}/2\hbar}~; \label{costate1}%
\end{equation}
its Wigner transform is the Gaussian%
\begin{equation}
W\phi_{0}(z)=(\pi\hbar)^{-n}e^{-|z|^{2}/\hbar}~. \label{wigcostate1}%
\end{equation}
More generally, the standard coherent state centered at $z_{0}=(x_{0},p_{0})$
is $\phi_{z_{0}}=\widehat{T}(z_{0})\phi_{0}$; explicitly
\[
\phi_{z_{0}}(x)=e^{-ip_{0}x_{0}/2\hbar}e^{ip_{0}x/\hbar}\phi_{0}(x-x_{0})
\]
and its Wigner transform is, using (\ref{cov2}),
\[
W\phi_{z_{0}}(z)=(\pi\hbar)^{-n}e^{-|z-z_{0}|^{2}/\hbar}~.
\]
More generally consider (normalized) generalized Gaussian states $|\psi
_{X,Y}\rangle$ with
\begin{equation}
\psi_{X,Y}(x)=\left(  \tfrac{1}{\pi\hbar}\right)  ^{n/4}(\det X)^{1/4}%
e^{-\tfrac{1}{2\hbar}Mx^{2}} \label{squeezed}%
\end{equation}
where $M=X+iY$, $X$ and $Y$ real symmetric, $X$ positive definite. These
states generalize the \textquotedblleft squeezed states\textquotedblright%
\ familiar from quantum optics, the \textquotedblleft squeezing
parameters\textquotedblright\ being the eigenvalues of $X$. The Wigner
transform of $\psi_{X,Y}$ is explicitly given by
\cite{Folland,Birk,Littlejohn}%
\begin{equation}
W\psi_{X,Y}(z)=\left(  \tfrac{1}{\pi\hbar}\right)  ^{n}e^{-\tfrac{1}{\hbar
}Gz^{2}} \label{ouipsi}%
\end{equation}
where $G$ is a symmetric \emph{symplectic} matrix:
\begin{equation}
G=R^{T}R\text{ \ , \ }R=%
\begin{pmatrix}
X^{1/2} & 0\\
X^{-1/2}Y & X^{-1/2}%
\end{pmatrix}
~. \label{gsym}%
\end{equation}
The \textquotedblleft covariance (or Wigner) ellipsoid\textquotedblright%
\ \cite{Littlejohn}
\[
\Omega_{\Sigma}=\{z:Gz^{2}\leq\hbar\}
\]
of $\psi_{X,Y}$ is the symplectic ball:%
\[
\Omega_{\Sigma}=S(B_{\sqrt{\hbar}}^{2n})\text{ \ },\text{ \ }S=R^{-1}%
\in\operatorname*{Sp}(n)~.
\]
Conversely, if an ellipsoid $\Omega$ is a symplectic ball $SB_{\sqrt{\hbar}%
}^{2n}$, then the function $\rho(z)=(\pi\hbar)^{-n}e^{-Gz^{2}/\hbar}$ with
$G=(S^{-1})^{T}S^{-1}$ is the Wigner transform of a Gaussian (\ref{squeezed})
up to an unessential prefactor with modulus one.

Let us consider more general phase space Gaussians%
\begin{equation}
\rho(z)=(\pi\hbar)^{-n}(\det M)^{1/2}e^{-\frac{1}{\hbar}Mz^{2}}~. \label{rhum}%
\end{equation}
Defining the covariance matrix of $\rho$ by
\begin{equation}
\Sigma=\frac{\hbar}{2}M^{-1} \label{covmatrix}%
\end{equation}
we can rewrite (\ref{rhum}) in the perhaps more familiar form%
\begin{equation}
\rho(z)=\frac{1}{(2\pi)^{n}\sqrt{\det\Sigma}}e^{-\frac{1}{2}\Sigma^{-1}z^{2}%
}~. \label{Gauss}%
\end{equation}
A fundamental result in harmonic analysis
\cite{Arvind,go03,Birk,Gauss1,Gauss2,Narcow89} is now that $\rho(z)$ is the
Wigner distribution of a mixed quantum state if and only if the
\textquotedblleft quantum condition\footnote{It is actually an equivalent form
of the Robertson--Schr\"{o}dinger inequalities \cite{FP,physreps}%
.}\textquotedblright%
\begin{equation}
\Sigma+\frac{i\hbar}{2}J\geq0\Longleftrightarrow M^{-1}+iJ\geq0
\label{quantum}%
\end{equation}
holds (\textquotedblleft$\geq0$\textquotedblright\ means \textquotedblleft is
positive semidefinite\textquotedblright; note that the eigenvalues of
$\Sigma+\frac{i\hbar}{2}J$ are real since $(iJ)^{\ast}=-iJ^{T}=iJ$). When
(\ref{quantum}) holds, the purity of the Gaussian state $\widehat{\rho}$ with
Wigner distribution (\ref{Gauss}) is (\cite{do} and \cite{Birk}, \S 9.3,
p.301)%
\begin{equation}
\mu(\widehat{\rho})=\left(  \tfrac{\hbar}{2}\right)  ^{n}(\det\Sigma
)^{-1/2}=\sqrt{\det M}~; \label{purity}%
\end{equation}
it follows that $\widehat{\rho}$ is a pure state $\psi_{M}$ if and only if
$\det M=1$. Defining the Wigner ellipsoid \cite{Littlejohn} of $\widehat{\rho
}$ by
\[
\Omega_{\Sigma}=\{z:\tfrac{1}{2}\Sigma^{-1}z^{2}\leq1\}=\{z:Mz^{2}\leq
\hbar\dot{\}}%
\]
we have the following important geometric reformulation of the quantum
condition (\ref{quantum}):

\begin{theorem}
[Wigner ellipsoid]\label{Thm4}The quantum condition $\Sigma+\frac{i\hbar}%
{2}J\geq0$ is satisfied if and only if $\Omega_{\Sigma}$ contains a symplectic
ball $S(B_{\sqrt{\hbar}}^{2n})$; equivalently
\begin{equation}
c(\Omega_{\Sigma})\geq\pi\hbar\label{capom}%
\end{equation}
where $c(\Omega_{\Sigma})$ is the symplectic capacity (\ref{capomega}) of the
covariance ellipsoid.
\end{theorem}

That the conditions $\Omega_{\Sigma}\supset S(B_{\sqrt{\hbar}}^{2n})$ for some
$S\in\operatorname*{Sp}(n)$ and $c(\Omega_{\Sigma})\geq\pi\hbar$ follows from
the formula (\ref{capomega}) which says that $c(\Omega_{\Sigma})=c_{\min
}^{\mathrm{lin}}(\Omega_{\Sigma})$. We have given a proof of these properties
in \cite{Birk,FP,physreps}; it makes use of the Williamson symplectic
diagonalization of $M$ and essentially consists in showing that the condition
(\ref{quantum}) is equivalent to the property that the symplectic eigenvalues
of $M=\frac{\hbar}{2}\Sigma^{-1}$ all are $\leq1$. We have called symplectic
balls of the type $S(B_{\sqrt{\hbar}}^{2n})$ \emph{quantum blobs}
\cite{blobs}; they appear in the Wigner formalism as minimum uncertainty phase
space ellipsoids. The condition (\ref{capom}) can be restated by saying that
$\Omega_{\Sigma}$ is the Wigner ellipsoid of a Gaussian state if and only if
contains a quantum blob.

Summarizing: there is a one-to-one correspondence between phase space
ellipsoids $\Omega_{\Sigma}$ satisfying $\operatorname*{Cap}(\Omega_{\Sigma
})\geq\pi\hbar$ and Gaussian states with Wigner ellipsoid $\Omega_{\Sigma}$.

\subsection{Time evolution of quantum subsystems}

Let $\widehat{\rho}$ be a density matrix on $\mathbb{R}^{n}$ with Wigner
distribution%
\[
\rho(z)=\int_{\mathbb{R}^{n}}e^{\tfrac{i}{\hbar}p\cdot y}\left\langle
x+\tfrac{1}{2}y|\widehat{\rho}|(x-\tfrac{1}{2}y\right\rangle dy~.
\]
We assume from now on that $\rho(z)$ is a Gaussian (\ref{rhum}), (\ref{Gauss})
and, returning to the notation $z=(z_{A},z_{B})$, we write $M=\frac{\hbar}%
{2}\Sigma^{-1}$ in block-form
\begin{equation}
M=%
\begin{pmatrix}
M_{AA} & M_{AB}\\
M_{BA} & M_{BB}%
\end{pmatrix}
~. \label{MC}%
\end{equation}
Since $M=M^{T}>0$ we have $M_{AA}^{T}=M_{AA}>0$, $M_{BB}^{T}=M_{BB}>0$, and
$M_{BA}^{T}=M_{AB}$. We define the Wigner distribution of the subsystem $A$ by
\textquotedblleft taking the partial trace\textquotedblright\
\begin{equation}
\rho_{A}(z_{A})=\int_{\mathbb{R}^{2n_{B}}}\rho(z_{A},z_{B})dz_{B}~.
\label{roadef}%
\end{equation}
A straightforward calculation of Gaussian integrals yields the formula%
\begin{equation}
\rho_{A}(z_{A})=(\pi\hbar)^{-n_{A}}(\det M/M_{BB})^{1/2}e^{-\frac{1}{\hbar
}(M/M_{BB})z_{A}^{2}} \label{rhoaza}%
\end{equation}
hence the covariance matrix of $\rho_{A}$ is
\begin{equation}
\Sigma_{A}=\frac{\hbar}{2}(M/M_{BB})^{-1} \label{sigmaa}%
\end{equation}
and the covariance ellipsoid of $\rho_{A}$ is thus
\begin{equation}
\Omega_{A}=\{z_{A}:(M/M_{BB})z_{A}^{2}\leq\hbar\}~; \label{covreda}%
\end{equation}
it is the orthogonal projection $\Pi_{A}\Omega$ on $\mathbb{R}^{2n_{A}}$ of
the covariance ellipsoid $\Omega$ of $\widehat{\rho}$. For $\rho_{A}(z_{A})$
to qualify as the Wigner distribution of a \textit{bona fide} partial mixed
state $\widehat{\rho}_{A}=\operatorname*{Tr}_{B}(\widehat{\rho})$ it still
remains to prove\footnote{This step is usually ignored in the literature. It
is needed to show that the partial trace indeed is a positive operator. It can
be proven directly using methods from functional analysis (the
\textquotedblleft Kastler--Loupias--Miracle-Sole conditions). Our approach
using the extended PSC is much simpler.} that $\Sigma_{A}$ satisfies the
quantum condition (\ref{quantum}). In view of Theorem \ref{Thm4} it suffices
for this to show that $\Omega_{A}$ contains a symplectic ball $S_{A}%
(B_{\sqrt{\hbar}}^{2n_{A}})$. But this is an immediate consequence of the
extended principle of the symplectic camel in Theorem \ref{Thm3}: in view of
the \textquotedblleft only if\textquotedblright\ part of Theorem \ref{Thm4}
the covariance ellipsoid $\Omega_{\Sigma}$ contains a symplectic ball
$S(B_{\sqrt{\hbar}}^{2n})$ in $\mathbb{R}^{2n}$ and hence $\Omega_{A}=\Pi
_{A}\Omega$ contains a symplectic ball $S_{A}(B_{\sqrt{\hbar}}^{2n_{A}})$ in
$\mathbb{R}^{2n_{A}}$.

We apply the results above to the motion of quantum subsystems. We begin by
noting that the Wigner transform directly links the approximate Hamiltonian
flow (\ref{phaseflow2})%
\[
U_{t}(z_{0})=T(z_{t})S_{t}(z_{0})T(-z_{0})
\]
to the corresponding semiclassical propagator (\ref{stu})%
\[
\widehat{U}_{t}(z_{0})=e^{\frac{i}{\hbar}\gamma(t)}\widehat{T}(z_{t}%
)\widehat{S}_{t}(z_{0})\widehat{T}(-z_{0})
\]
via the intertwining formulas (\ref{cov1}), (\ref{cov2}). We have:
\begin{equation}
W(\widehat{U}_{t}(z_{0})\psi)(z)=W\psi(U_{t}(z_{0})^{-1}(z))~. \label{wut}%
\end{equation}
Here is a simple proof of this equality. Since the prefactor $e^{\frac
{i}{\hbar}\gamma(t)}$ in (\ref{stu}) is eliminated by complex conjugation in
the definition (\ref{wimo}) of the Wigner transform we have, using several
times (\ref{cov1}), (\ref{cov2}),
\begin{align*}
W(\widehat{U}_{t}(z_{0})\psi)(z)  &  =W(\widehat{T}(z_{t})\widehat{S}%
_{t}(z_{0})\widehat{T}(-z_{0})\psi)(z)\\
&  =W(\widehat{S}_{t}(z_{0})\widehat{T}(-z_{0})\psi)(z-z_{t})\\
&  =W(\widehat{T}(-z_{0})\psi)(S_{t}(z_{0})^{-1}(z-z_{t}))\\
&  =W\psi(S_{t}(z_{0})^{-1}(z-z_{t})+z_{0})~.
\end{align*}
Now, for any $z\in\mathbb{R}^{2n}$,
\begin{align*}
U_{t}(z_{0})^{-1}z  &  =T(z_{0})S_{t}(z_{0})^{-1}T(-z_{t})z\\
&  =S_{t}(z_{0})^{-1}(z-z_{t})+z_{0}%
\end{align*}
from which the equality in (\ref{wut}) follows.

Let us state and prove our main result. We write $P_{t}(z_{0})=(S_{t}%
(z_{0})S_{t}^{T}(z_{0}))^{-1}$ in block matrix form (\ref{blockp}).

\begin{theorem}
\label{ThmSemi}Assume that the bipartite quantum system $A\cup B$ is in the
Gaussian state $|\phi_{z_{0}}\rangle$ at initial time $t=0$:
\[
\phi_{z_{0}}(x)=(\pi\hbar)^{-n/4}e^{ip_{0}x/\hbar}e^{-|x-x_{0}|^{2}/2\hbar}~.
\]
At time $t$ the subsystem $A$ will be in a Gaussian mixed state $\widehat{\rho
}_{A,t}$ with Wigner distribution%
\begin{equation}
\rho_{A,t}(z)=(\pi\hbar)^{-n_{A}}(\det M_{A,t}(z_{0}))^{1/2}e^{-\frac{1}%
{\hbar}M_{A,t}(z_{0})z_{a}^{2}} \label{mat}%
\end{equation}
with \ \
\[
M_{A,t}(z_{0})=P_{t}(z_{0})/P_{BB,t}(z_{0}\dot{)}~.
\]
The purity of the state $\widehat{\rho}_{A,t}$ is
\begin{equation}
\mu(\widehat{\rho}_{A,t})=\frac{1}{\sqrt{\det P_{BB,t}(z_{0}\dot{)}}}~.
\label{purityproj}%
\end{equation}

\end{theorem}

\begin{proof}
It is sufficient to study the case $z_{0}=0$ since the general case is
obtained by translations. To simplify notation we write $U_{t}$, $P_{t}$,
$S_{t}$, \textit{etc.} instead of $U_{t}(0)$, $P_{t}(0)$, $S_{t}(0)$,... The
classical and semiclassical evolution operators are thus here
\[
U_{t}=T(z_{t})+S_{t}\text{ , }\widehat{U}_{t}=\widehat{T}(z_{t})+\widehat{S}%
_{t}~.\text{ }%
\]
In the geometric phase space picture $\phi_{z_{0}}=\phi_{0}$ is represented by
the phase space ball $B_{\sqrt{\hbar}}^{2n}$, and we have%
\[
U_{t}=\{z_{t}\}+S_{t}(B_{\sqrt{\hbar}}^{2n})
\]
to which corresponds the pure Gaussian state $\widehat{U}_{t}\phi_{0}$. In
view of formula (\ref{wut}) we have
\begin{equation}
W(\widehat{U}_{t}\phi_{0})(z)=W\phi_{0}(U_{t}{}^{-1}z)~. \label{wut0}%
\end{equation}
In view of the extended principle of the symplectic camel, the evolution of
the orthogonal projection
\[
\Omega_{A,t}=\Pi_{A}(U_{t}(B_{\sqrt{\hbar}}^{2n}))
\]
on the partial phase space $\mathbb{R}^{2n_{A}}$ satisfies
\begin{equation}
\Omega_{A,t}\supset\{z_{A,t}\}+S_{A,t}(B_{\sqrt{\hbar}}^{2n_{A}})
\end{equation}
(formula (\ref{satt}) in Theorem \ref{ThmFund}). It is explicitly given by%
\[
\Omega_{A,t}=\{z_{A}:(P_{t}/P_{BB,t})z_{A}^{2}\leq\hbar\}
\]
where
\[
P_{t}/P_{BB,t}=(S_{A,t}^{-1})^{T}D_{A,t}S_{A,t}^{-1}%
\]
is the Schur complement of $P_{BB,t}$ in $P_{t}=(S_{t}S_{t}^{T})^{-1}$. In
view of formula (\ref{rhoaza}) $\Omega_{A,t}$ is the covariance ellipsoid of a
mixed quantum state with Wigner distribution
\[
\rho_{A}(z_{A})=(\pi\hbar)^{-n_{A}}(\det(P_{t}/P_{BB,t}))^{1/2}e^{-\frac
{1}{\hbar}(P_{t}/P_{BB,t})z_{A}^{2}}~
\]
hence (\ref{mat}). In view of formula (\ref{purity}) the purity of this state
is
\[
\mu(\widehat{\rho}_{A,t})=\sqrt{\det(P_{t}/P_{BB,t})}%
\]
whence (\ref{purityproj}) in view of formula (\ref{Schur3}).
\end{proof}

Comparing formulas (\ref{entropy}) and (\ref{purityproj}) we see that the
variations of classical entropy and mixedness are related by
\begin{equation}
\Delta\mathcal{S}=-2k_{\mathrm{B}}\ln\mu(\widehat{\rho}_{A,t})~.
\label{eltamu}%
\end{equation}
As in Section \ref{secentropy} (formula (\ref{volate})) the mixedness of the
projected state increases due to its interaction with the subsystem $B$. It
remains constant if and only if $\widehat{\rho}_{A,t}$ is a pure state, which
requires, as in the classical case, that $P_{t}(z_{0})=P_{AA,t}(z_{0})\oplus
P_{BB,t}(z_{0})$ which means that the subsystems $A$ and $B$ do not interact.

\section{Perspectives and Speculations\label{secpersp}}

All our results for subsystems (both classical and semiclassical) crucially
depend on the generalization in Theorem \ref{Thm3} to the linear case of
Theorem \ref{Thm2} (the principle of the symplectic camel). A natural question
that arises is whether one could extend Theorem \ref{Thm3} to more general
canonical transformations than the linear (or affine) ones. For instance, if
$(\Phi_{t}^{H})$ is the phase-flow determined by a Hamiltonian of the
classical type
\[
H(x,p,t)=\frac{1}{2}m^{-1}p^{2}+V(x,t)
\]
could it be true that the projections of $\Phi_{t}(B_{R}^{2n})$ onto the
subspaces $\mathbb{R}^{2n_{A}}$ and $\mathbb{R}^{2n_{B}}$ contain images of
the balls $B_{R}^{2n_{A}}$ and $B_{R}^{2n_{B}}$ by canonical transformations
of $\mathbb{R}^{2n_{A}}$ and $\mathbb{R}^{2n_{B}}$? When $n_{A}=1$ this is
trivially true in view of Gromov's theorem: the projection of $\Phi_{t}%
(B_{R}^{2n})$ onto the plane $\mathbb{R}^{2}$ has an area of at least $\pi
R^{2}$ and must therefore contain the image of the disk $B_{R}^{2}$ by an
area-preserving diffeomorphism of the plane and such diffeomorphisms are
automatically canonical. In the general case $n_{A}>1$ the problem is open at
the time of writing; in fact Abbondandolo and Matveyev \cite{abbo} have shown
that there exist Hamiltonian flows $\Phi_{t}$ for which the answer is
negative, but the associated Hamiltonians are very unphysical. On the positive
side, as already mentioned above, Abbondandolo and Benedetti \cite{abbonew}
have recently refined the results in \cite{abbo} and shown that if the
$\Phi_{t}$ are sufficiently close to linear canonical transformations, then
the volume inequality (\ref{abbo1}) holds, which is a weaker statement than
Theorem \ref{Thm3}. On the other hand, in view of Theorem \ref{Thm3} one may
suspect that the Gromov width of the projection of the image of the of the
ball $B_{R}^{2n}$ by a canonical transformation that is close to a linear one
should be at least $\pi R^{2}$. Even if it is hard to see why such properties
should not be true, we are lacking, for the time being, mathematical
justifications; the above mentioned advances are highly qualitative and seem
to be difficult to implement in practice. They are all related to the question
whether the Gromov width increases under symplectic projections. They
certainly deserve to be studied further.

Another topic which might be worth exploring using the methods outlined in
this paper is the study of Poincar\'{e} recurrence for subsystems. As we have
explained elsewhere \cite{Entropy1} the notion of symplectic capacity seems to
play a fundamental role in recurrence (it was one of the motivations of Gromov
in his study \cite{Gromov} of symplectic non-squeezing properties; see Schlenk
\cite{Schlenk}). It is clear from Theorem \ref{Thm3} that recurrence in the
subsystems $A$ and $B$ is liable to occur faster than in the total system
$A\cup B$. It would be interesting to study this property in relation with the
entropy briefly discussed in Section \ref{secentropy}. We will come back to
this question in a near future.

\begin{acknowledgement}
This work has been supported by the grant P33447 of the Austrian Research
Agency FWF.
\end{acknowledgement}

\begin{acknowledgement}
It is my pleasure and my duty to thank Glen Dennis for a careful reading of
the manuscript and for having pointed out various typos. I also express my
gratitude to the Reviewer for useful comments about the \textquotedblleft
symplectic camel\textquotedblright.
\end{acknowledgement}

\textbf{Data availability statement:} Data sharing not applicable -- no new
data generated

\end{document}